\documentclass{article} 

\usepackage{color}
\usepackage[pdftex,pagebackref=false]{hyperref}
\definecolor{darkblue}{rgb}{0,0.08,0.45} 
\hypersetup{colorlinks,citecolor=darkblue,linkcolor=darkblue,urlcolor=darkblue}

\usepackage{amssymb,amsthm,amsmath}
\usepackage[pdftex]{graphicx}
\usepackage{graphics}
\usepackage{epsfig}
\input{epsf.sty}
\usepackage{color}
\usepackage{fullpage}

\usepackage{pgf,tikz}
\usetikzlibrary{arrows}
\usepackage{ifthen}

\newtheorem{lemma}{Lemma}

\newtheorem{theorem}{Theorem}
\newtheorem{proposition}{Proposition}
\newtheorem{claim}{Claim}
\newtheorem*{theorem*}{Theorem}
\newtheorem*{corollary*}{Corollary}

\makeatletter
\newtheorem*{rep@theorem}{\rep@title}
\newcommand{\newreptheorem}[2]{%
\newenvironment{rep#1}[1]{%
 \def\rep@title{#2 \ref{##1}}%
 \begin{rep@theorem}}%
 {\end{rep@theorem}}}
\makeatother

\newreptheorem{theorem}{Theorem}
\newreptheorem{lemma}{Lemma}
\newreptheorem{proposition}{Proposition}

\newcommand{\dist}{d}

\newcommand{\ceil}[1]{\lceil #1 \rceil}
\newcommand{\floor}[1]{\lfloor #1 \rfloor}
\newcommand{\abs}[1]{\left| #1 \right|}
\newcommand{\set}[1]{\left\{ #1 \right\}}
\newcommand{\comment}[1]{}

\newcommand{\lgO}{\tilde O}
\newcommand{\sncn}{\mathcal{SNC}(n)}

\newcommand\querytime{O(\epsilon^{-1}\log^2(1/\epsilon)\log\log(1/\epsilon)\log^*(n)+\log\log\log(n))}
\newcommand\querytimepoly{O(\epsilon^{-1}\log^2(1/\epsilon)\log\log(1/\epsilon)\log^*(n))}
\newcommand\querytimeabs{\lgO(\epsilon^{-1}+\log\log\log n)}
\newcommand\querytimeabspoly{\lgO(\epsilon^{-1})}
\newcommand\spacecomplexity{O(n\left[\log\log(n)+\log(1/\epsilon)\right]\log^*(n)\log\log(1/\epsilon))}
\newcommand\spacecomplexitypoly{O(n\left[\log(n)+\log(1/\epsilon)\right]\log^*(n)\log\log(1/\epsilon))}
\newcommand\spacecomplexityabs{\lgO(n\log\log n)}
\newcommand\spacecomplexityabspoly{\lgO(n\log n)}

\newcommand\spacecomplexityconsteps{O(n\log\log(n)\log^*(n))}
\newcommand\querytimeconsteps{O(\log\log\log n)}

\newcommand\querytimeunw{O(\epsilon^{-3}\log^*(n))}
\newcommand\spacecomplexityunw{2^{O(\log^*n)}n\epsilon^{-1}}

\makeindex

\begin{document}
\begin{titlepage}
\title{More Compact Oracles for\\Approximate Distances in Planar Graphs}
\author{Christian Sommer\\\url{csom@mit.edu}}

\maketitle
\thispagestyle{empty}
\setcounter{footnote}{0}

\begin{abstract} 
Distance oracles are data structures that provide fast (possibly approximate) answers to shortest-path and distance queries in graphs. 
The tradeoff between the space requirements and the query time of distance oracles is of particular interest and the main focus of this paper. 

In FOCS`01, Thorup introduced approximate distance oracles for planar graphs. 
He proved that, for any $\epsilon>0$ and for any planar graph on $n$ nodes, there exists a $(1+\epsilon)$--approximate distance oracle 
using space $O(n\epsilon^{-1}\log n)$ such that approximate distance queries can be answered in time $O(\epsilon^{-1})$. 

Ten years later, we give the first improvements on the space--query time tradeoff for planar graphs. 

\begin{itemize}
\item We give the first oracle having a space--time product with subquadratic dependency on~$1/\epsilon$. 
For space $\spacecomplexityabspoly$ we obtain query time  $\querytimeabspoly$ (assuming polynomial edge weights). 
We believe that the dependency on $\epsilon$ may be almost optimal.

\item For the case of {\em moderate} edge weights (average bounded by $poly(\log n)$, which appears to be the case for many real-world road networks),  
we hit a ``sweet spot,'' improving upon Thorup's oracle both in terms of $\epsilon$ and~$n$. 
Our oracle uses space $\spacecomplexityabs$ and it has query time $\querytimeabs$. 
\end{itemize}

(Notation: $\lgO(\cdot)$ hides low-degree polynomials in $\log(1/\epsilon)$ and  $\log^*(n)$.)

\end{abstract}
\end{titlepage}

\section{Introduction}

{\em Distance oracles}~\cite{ThorupZwick2005} generalize the all-pairs shortest paths problem as follows: 
instead of computing and storing a distance matrix (with quadratic space requirements and pairwise 
distance computations that require {\em one} table look-up only), we wish to compute a data structure 
that requires sub-quadratic space $S$ but still allows for {\em efficient} (as in sublinear query time $Q$) 
distance computations. Depending on the application, it may be acceptable to output {\em approximate} 
answers to shortest-path and distance queries. The estimate provided by the distance oracle is supposed to be {\em at least} as large as the actual distance. 
The {\em stretch} $\alpha\geq1$ of an approximate distance oracle is defined as the worst-case 
ratio over all pairs of nodes of the query result divided by the actual shortest-path length.

Distance oracles can potentially be used in applications such as 
route planning and navigation~\cite{conf/sofsem/Goldberg07,reference/algo/Zaroliagis,RoutePlanningSurvey}, 
Geographic Information Systems (GIS) and intelligent transportation systems~\cite{conf/cikm/JingHR96}, 
logistics, traffic simulations~\cite{Ziliaskopoulos:Kotzinos:1997,conf/esa/BarrettBJKM02,journals/fgcs/RaneyN04,SPAcceleration}, 
computer games~\cite{SmartMove}, 
server selection~\cite{NgZhang,1015471,conf/icdcs/CostaCRK04,Shavitt:2008:HEI:1373452.1373455,conf/infocom/ErikssonBN09}, 
XML indexing~\cite{conf/edbt/SchenkelTW04,conf/icde/SchenkelTW05}, 
reachability in object databases, 
packet routing~\cite{1094690}, 
causal regulatory networks~\cite{CausalRegulatoryNetworks}, and 
path finding in social networks~\cite{Karinthy,Milgram,ScientificCollaborationNetworks}.

For general graphs, distance oracles use large amounts of space, or they have long query time, or their 
stretch is at least two~\cite{ThorupZwick2005,SparseDO,PatrascuRoditty}. 
In this work we consider planar graphs, for which the known tradeoffs between stretch, space, and 
query time are much better (see Table~\ref{tab:planardoresults} for an overview). 

One important reason for this better tradeoff performance is that algorithms 
can make use of small {\em separators}~\cite{Ungar01101951,LT79,journals/jcss/Miller86}. For approximate distance 
oracles, separators consisting itself of a small number of shortest paths are particularly useful~\cite{ThorupJACM04,conf/soda/Klein02}.
Separator-based approaches however 
often use recursion of logarithmic depth, which manifests itself by logarithmic factors in either the space requirement 
or the query time (or even both). 
Currently, the best tradeoff is provided by Thorup's approximate distance oracle, for which 
we provide a brief technical outline in Sections~\ref{sec:planarsep} and~\ref{sec:prelim:thorup}. 
The best exact distance oracles~\cite{journals/jcss/FakcharoenpholR06,DjidjevWG96,stoc/ChenX00,conf/soda/Cabello06,Nussbaum11,MS12}  
have a space-query time product $S\cdot Q$ proportional to roughly $n\sqrt{n}$. For the 
best $(1+\epsilon)$--approximate distance oracle~\cite{ThorupJACM04}, that product $S\cdot Q$ is $O(n\log{n})$ for constant $\epsilon$. 

\begin{table*}[h]
\begin{center}
\begin{tabular}{|l  | l | l | l | l|}
\hline
Space & Query & Stretch & Reference \\
\hline
\hline
$O(n)$ & $O(n)$ & $1$ & SSSP~\cite{journals/jcss/HenzingerKRS97}\\
$O(n^2)$ & $O(1)$ & $1$ & APSP\\
$O(n^2\log\log(n)/\log(n))$ & $O(1)$ & 1 & \cite{WNSubQ10}\\
$O(S)$  & $O(nS^{-1/2}\log^{2}(n)(\log\log n)^{3/2})$ & $1$ & \cite{MS12} for any $S\in[n\log\log n,n^2]$\\
$O(n\log n)$ & $O(\sqrt n\log^{5/2}n)$ & $1$ & \cite{journals/jcss/FakcharoenpholR06}\\
$O(n)$ & $O(n^{(1/2)+\delta})$ & $1$ & \cite{MS12} for any constant $\delta>0$\\
\hline
$O(n\epsilon^{-1}\log n)$           & $O(\epsilon^{-1})$& $1+\epsilon$ & \cite{ThorupJACM04}\\
$O(n)$               & $O(\epsilon^{-2}\log^2 n)$ & $1+\epsilon$ & \cite{KKS}\\
\hline
\hline
$\spacecomplexityabspoly$ & $\querytimeabspoly$ &$1+\epsilon$ & Theorem~\ref{thm:polyweights} (up to $\log(1/\epsilon)$ and $\log^*(n)$)\\
$\spacecomplexityabs$ & $\querytimeabs$ &$1+\epsilon$ & Theorem~\ref{thm:eps} (up to $\log(1/\epsilon)$ and $\log^*(n)$)\\
\hline
$O(n\log\log n)$ & $O(\log\log\log n)$ &$O(1)$ & Proposition~\ref{prop:const}\\
\hline
\end{tabular}
\end{center}
\caption{Time and space complexities of distance oracles for undirected planar graphs (some results extend to planar digraphs).
Existing results are for arbitrarily weighted graphs; our new results are for graphs with polynomial weights (Theorem~\ref{thm:polyweights}) and for 
 {\em moderately weighted} graphs (Theorem~\ref{thm:eps}) only.  $\log^*(n)$ is the {\em iterated logarithm} of $n$. 
For the bounds corresponding to our results, $\lgO(\cdot)$ hides low-degree polynomials in $\log(1/\epsilon)$ and $\log^*(n)$. \label{tab:planardoresults}}
\end{table*}

Our main result is a distance oracle that improves upon these tradeoffs in terms of {\em both} $n$ {\em and} $\epsilon$ for graphs with {\em moderate} edge weights.
By moderate we mean that, after normalization such that the smallest edge weight is 1, the average weight is bounded by $poly(\log n)$. 
We believe that this is a reasonable assumption as the average weight for the European road network (the version made available for scientific use by the company PTV AG) 
appears to be not too large.\footnote{The network covers European 14 countries and it has 18,010,173 nodes and 42,560,279 edges. 
It serves as an important benchmark graph for shortest-path query methods~\cite{DIMACSChallenge}. 
For the travel time metric ({\tt scc-eur2time.gr}), the total weight is 21,340,824,356, which yields an average of approximately 501. For the distance metric ({\tt scc-eur2dist.gr}) the two 
values are 9,420,195,951 (sum) and 221 (rounded average). Also, $\log_2(42560279)\approx25.34$ and thus $(\log_2\abs{E})^2\geq501$. 
We do not claim to rigorously distinguish $O(poly(\log\abs{E}))$ from $\Omega(\sqrt{\abs{E}})$ for this $E$ (here $\sqrt{42560279}\approx6 523.82$). }
We provide a $(1+\epsilon)$--approximate distance oracle with the following characteristics: for constant $\epsilon$, the 
space is $S=\spacecomplexityconsteps$ and the query time is $Q=\querytimeconsteps$, thus $S\cdot Q=o(n\log n)$. 
We show how the logarithmic dependency on~$n$, which is quite common in separator-based approaches, can be avoided entirely, 
which may be interesting in its own right.\footnote{Recently, the complexity of some optimization problems for planar graphs has been improved~\cite{nlglgn-mincut-WN10} 
from worst-case running time $O(n\log n)$ to $O(n\log\log n)$ using fast {\em $r$--divisions}~\cite{journals/siamcomp/Frederickson87} (and other techniques). 
For approximate distance oracles, a similar approach improves the space requirements at the cost of the query time~\cite{KKS}. 
In this work, we improve the {\em tradeoff}, i.e.~the product of space and query time.} 
Equally important, our construction also yields the first space--time product with {\em subquadratic} dependency on~$\epsilon$
 (known constructions have $\epsilon^{-1}$ in both space and query~\cite{ThorupJACM04,conf/soda/Klein02} or $\epsilon^{-2}$ in the query complexity~\cite{KKS}). 

More precisely, we prove the following. All our bounds are in the {\em word RAM model}~\cite{journals/jcss/CookR73}.
\def\mainthm{For any undirected planar graph $G$ on $n$ nodes with average weight $\leq\log^\theta n$ for some constant $\theta\geq0$, 
and for any $\epsilon>0$ there exists a $(1+\epsilon)$--approximate distance oracle with query time
$\querytime$ using space $\spacecomplexity$, and preprocessing time $O(n\epsilon^{-2}\log^3(n)\log\log(n)\log^*(n)+\sncn\cdot\log\log n)$, where $\sncn$ denotes the time required to compute a sparse neighborhood cover.}

\def\mainthmpoly{For any undirected planar graph $G$ on $n$ nodes with edge weights polynomial in $n$ 
and for any $\epsilon>0$ there exists a $(1+\epsilon)$--approximate distance oracle with query time
$\querytimepoly$ using space $\spacecomplexitypoly$, and preprocessing time $O(n\epsilon^{-2}\log^4(n)\log^*(n)+\sncn\cdot\log n)$, where $\sncn$ denotes the time required to compute a sparse neighborhood cover.}

\def\mainthmunw{For any unweighted, undirected planar graph $G$ on $n$ nodes, 
and for any $\epsilon>0$ there exists a $(1+\epsilon)$--approximate distance oracle with query time
$\querytimeunw$ using space $\spacecomplexityunw$, and preprocessing time $O(...+\sncn\cdot\log\log n)$, where $\sncn$ denotes the time required to compute a sparse neighborhood cover.}

\begin{theorem}
\mainthmpoly
\label{thm:polyweights}
\end{theorem}

\begin{theorem}
\mainthm
\label{thm:eps}
\end{theorem}

As a starting point, we provide the following distance oracle with constant stretch. 

\def\constthm{For any undirected planar graph $G$ on $n$ nodes with average weight bounded by $\log^\theta n$ for some constant $\theta\geq0$ there exists an $O(1)$--approximate distance oracle with query time
$O(\log\log\log n)$ using space $O(n\log\log n)$, and preprocessing time $O(n\log^3n+\sncn\cdot\log\log n)$, where $\sncn$ denotes the time required to compute a sparse neighborhood cover.}

\def\linspacethm{For any unweighted, undirected planar graph $G$  and for any $\epsilon>0$ and $\delta>0$ 
there exists a $(1+\epsilon)$--approximate distance oracle with query time
$O(\epsilon^{-2}\log^{1+\delta}(n))$ using space $O(n)$, and preprocessing time $O(...+\sncn\cdot\log\log n)$, where $\sncn$ denotes the time required to compute a sparse neighborhood cover..}

\begin{proposition}
\constthm
\label{prop:const}
\end{proposition}

Somewhat surprisingly, nothing more efficient than Thorup's $(1+\epsilon)$--stretch oracle had been known for arbitrary constant approximations. 
The only $O(1)$--stretch (as opposed to $(1+\epsilon)$--stretch) oracle constructions for planar graphs we are aware of are 
\cite{conf/soda/Chen95,ArikatiCCDSZ96,conf/focs/GuptaKR01} and the following two indirect constructions: 
one construction is by using a routing scheme~\cite{journals/siamcomp/FredericksonJ89,journals/siamcomp/FredericksonJ90} and 
another construction is by using an $\ell_\infty$--embedding~\cite{conf/focs/KrauthgamerLMN04}. 
The space--query time tradeoff of our data structure is better than that of all the previous constructions.

\paragraph{Techniques}Our data structures make extensive use of Thorup's approximate distance oracle~\cite{ThorupJACM04} and of a large variety of techniques developed for planar graphs in recent years. 
We also introduce a new kind of {\em $r$--divisions}~\cite{journals/siamcomp/Frederickson87} based on shortest-path separators. Other techniques we use 
include {\em crossing substitutes}~\cite{journals/algorithmica/KleinS98}, {\em sparse neighborhood covers}~\cite{conf/podc/BuschLT07}, generalized {\em dominating sets}~\cite{KP98}, 
and the {\em Cycle MSSP} data structure~\cite{MS12}. 

\section{Preliminaries}

We use standard terminology from graph theory, see for example~\cite{DiestelBook}. Graphs we consider are planar, undirected, unweighted 
(if not explicitly mentioned otherwise --- the extension to moderately weighted graphs is emphasized in the pseudocode), and they have $n$ nodes. 

Let $\log^*(\cdot)$ denote the {\em iterated logarithm} function (also {\em super-logarithm}), which is defined as $\log^*n=1+\log^*(\log n)$ for $n>1$ and as $0$ for $n\leq1$.

\subsection{Planar Separators and $r$--divisions}
\label{sec:planarsep}
A {\em separator} for a graph $G=(V,E)$ is a subset of the nodes $S\subseteq V$ such that 
removing the nodes in $S$ from $G$ partitions the graph into at least two disconnected components. 
Let us assign a {\em weight} $w\in[0,1]$ to every node $v\in V$. 
A separator is deemed {\em balanced} if none of the resulting components has weight more 
than a constant fraction $\rho$ of the total weight for some constant $\rho<1$. 

Planar graphs are known to have {\em small} separators: 
for any planar graph there exists a balanced separator 
consisting of $O(\sqrt n)$ nodes~\cite{Ungar01101951,LT79,journals/jcss/Miller86}. 
Recursively separating a graph into smaller components, 
we obtain a {\em division}  into edge-induced subgraphs. 
A node of $G$ is a {\em boundary node}  of the partition if it belongs to more than one subgraph. 
An {\em $r$--division}~\cite{journals/siamcomp/Frederickson87} partitions $G$ into $O(n/r)$ subgraphs, 
called {\em regions}, each consisting of  $O(r)$ edges with at most $O(\sqrt{r})$ nodes on the boundary. 

Separators can be chosen, for example, to form a cycle~\cite{journals/jcss/Miller86} 
or also to form a set of paths~\cite{LT79,ThorupJACM04}. 
Lipton and Tarjan~\cite{LT79} prove that, for any 
spanning tree $T$ in a triangulated planar graph,
there is a non-tree edge $e$ such that the unique 
simple cycle in $T\cup \set{e}$ is a balanced separator (fundamental cycle separator).  
Thorup~\cite{ThorupJACM04} uses this construction with a {\em shortest-path tree} $T$, 
rooted at an arbitrary source node. For any node $u$, let $T(u)$ denote the tree path from 
$u$ to the root. 
\begin{lemma}[{Shortest-Path Separability; Thorup~\cite[Lemma~2.3]{ThorupJACM04}}]
In linear time, given an undirected planar graph $H$ with a rooted
spanning tree $T$ and non-negative vertex weights, we can find three vertices $u$, $v$,
and $w$ such that each component of $H\setminus V(T (u) \cup T (v) \cup T (w))$ has at most half the weight of~$H$.
\label{lem:thorupseparator}
\end{lemma}
Since $T$ is a shortest-path tree, a separator $S$ consists of at most three {\em shortest} paths.

\subsection{Approximate Distance Oracle and Labeling Scheme}
\label{sec:prelim:thorup}
The approximate distance oracle for planar graphs by Thorup~\cite{ThorupJACM04} 
can be distributed as a distance labeling scheme~\cite{PelegLabeling00,conf/esa/GavoilleKKPP01,journals/jal/GavoillePPR04}. 
Each node $u$ is assigned a label $\mathcal L(u)$ such that there is a decoding function $\mathcal D(\cdot,\cdot)$ 
that approximates the distance between $u$ and $v$ based on their labels only. 

\begin{lemma}[{Thorup~\cite[Theorem~3.16 and Theorem~3.19]{ThorupJACM04}}]
There is an algorithm that computes $(1+\epsilon)$--approximate distance labels for an $n$--node planar graph with the following properties. 
The algorithm runs in time $O(n\epsilon^{-2}\log^3n)$ and outputs labels of length 
 $O(\epsilon^{-1}\log n)$ with query time $O(\epsilon^{-1})$.
\label{lemma:thoruplabel}
\end{lemma}

One of the key ideas in Thorup's oracle is the concept of {\em shortest-path tree separability} 
as outlined in the previous section (Lemma~\ref{lem:thorupseparator}). 
A constant number of shortest paths $Q$ separate the graph into components of at most
half the size. Another key idea is to approximate shortest $s-t$ paths that intersect a separator path~$Q$. 
Let the shortest $s-t$ path intersect $Q$ at a particular node $q\in Q$. If we are willing to accept 
a slightly longer path, we can restrict the number of possible intersections from $\abs{Q}$ 
to $O(1/\epsilon)$, for $\epsilon>0$. 
We use the following variant of Klein and Subramanian~\cite{journals/algorithmica/KleinS98} (which they also call a {\em crossing substitute}). 
\begin{lemma}[{Klein and Subramanian~\cite[Lemma~4]{journals/algorithmica/KleinS98}}]
Let $Q$ be a path of length $O(D)$ for an integer $D$. 
Let $C\subseteq Q$ be a set of $O(1/\epsilon)$ equally spaced nodes on $Q$, called {\em the $\epsilon$--cover of $Q$}. 
Any pair of nodes $(u,v)$ at distance $[D,2D)$ whose shortest path intersects $Q$ can be {\em $(1+\epsilon)$--approximated} by a path going through a node $c\in C$, that is, 
\begin{displaymath}
\dist(u,v)\leq\min\limits_{c\in\mathcal C}\dist(u,c)+\dist(c,v)\leq (1+\epsilon)\dist(u,v).
\end{displaymath}
\label{lem:kleinsubramanian}
\end{lemma}
The proof is based on the observation that any {\em detour} using node $c\in C$ instead of node $q$ 
could cause an additive error of at most $O(\epsilon D)$.

\subsection{Distance--$\delta$ Dominating Sets}
Let $\delta < n$ be an integer.  A {\em $\delta$--dominating set} of a graph $G = (V,E)$ is a subset $L \subseteq V$ of nodes such that for each $v\in V$ there is a node $l\in L$ at distance at most $\delta$. 
It is well-known that there is a $\delta$--dominating set $L$ of size at most $\abs{L}\leq n/(\delta+1)$ and that such a set $L$ can be found efficiently~\cite{KP98}. 

A related but hard (even for degree--3 planar graphs) problem is the {\em minimum $p$--center problem}~\cite{PlesnikCentersComplexity,DyerFrieze85,PlesnikPCenters}, where we want to compute a 
set $U\subseteq V$ of size $p=\abs{U}$ such that the maximum distance from any node $v\in V$ to $U$ is minimized. We are mainly interested in that worst-case distance, for which the 
best guarantee is $n/p$.

\subsection{Sparse Neighborhood Covers}
\label{sec:sparsecover}
Busch, LaFortune, and Tirthapura~\cite{conf/podc/BuschLT07} provide 
{\em sparse covers}~\cite{conf/focs/AwerbuchP90a,journals/siamcomp/AwerbuchBCP98,conf/spaa/AbrahamGMW07,conf/podc/BuschLT07} for planar graphs. 
\begin{lemma}[{Busch, LaFortune, and Tirthapura~\cite{conf/podc/BuschLT07}}]
For any planar graph $G$ and for any integer $r$,  there is a {\em sparse cover}, which is a collection of connected subgraphs $(G_1,G_2,\dots)$, with the following properties: 
\begin{itemize}
\item for each node $v$ there is at least one subgraph $G_i$ that contains all neighbors within distance $r$, 
\item each node $v$ is contained in at most $30$ subgraphs, and 
\item each subgraph has {\em radius} at most $\rho=24r-8$\\(a graph has radius $\rho$ if it contains a spanning tree of depth $\rho$). 
\end{itemize}
Furthermore, such a sparse cover can be computed in polynomial time. 
\label{lemma:sparsecover}
\end{lemma}

Note that, since each node is in at most $O(1)$ subgraphs $G_i$, the total size of the cover is~$O(n)$. 

While for minor-free graphs the best-known construction algorithm requires polynomial time, their algorithm 
for planar graphs appears to actually run in time $O(n\log n)$. Let $\sncn$ denote the time required to compute sparse neighborhood covers. 
We state our preprocessing bounds with respect to $\sncn$.

\subsection{Exact Distance Oracles and Cycle MSSP}

We use the Cycle Multiple-Source Shortest Paths (MSSP) data structure~\cite{MS12}, which is a variant of Klein's MSSP data structure~\cite{conf/soda/Klein05}.

\begin{lemma}[{Cycle MSSP~\cite[Theorem~4]{MS12}}]
Given a directed planar graph $G$
on $n$ nodes and a simple cycle $C$ with $c = O(\sqrt n)$ nodes, there is an algorithm that preprocesses $G$
in $O(n \log^3 n)$ time to produce a data structure of size $O(n \log\log c)$
that can answer the following queries in $O(c \log^2c \log\log c)$ time:
for a query node $u$, output the distance from $u$ to all the nodes of $C$.
\label{lemma:cyclemssp}
\end{lemma}

We also use an exact distance oracle for small subgraphs. 

\begin{lemma}[{\cite[Theorem~3]{MS12}}]
For any directed planar graph $G$ with non-negative arc lengths, 
there is a data structure that supports 
exact distance queries in $G$ with the following properties: 
the preprocessing time is $O(n\log(n)\log\log(n))$, 
the space required is $O(n\log\log n)$, and 
the query time is $O(\sqrt n\log^2(n)\log\log(n))$. 
\label{lemma:exactdq}
\end{lemma}

\section{A More Compact Distance Oracle with Constant Stretch}
\label{sec:constantapprox}
The overall structure of the distance oracle presented in this section is also reused for the $(1+\epsilon)$--approximate distance oracle.  
The basic idea is to use Thorup's distance labels for long-range distances, which are 
the distances of length at least roughly $\log n$, and to use sparse neighborhood covers (Section~\ref{sec:sparsecover}) for $\log\log n$ different levels
to approximate short-range distances. 

We prove Proposition~\ref{prop:const}, which we restate here. 
\begin{repproposition}{prop:const}
\constthm
\end{repproposition}

\subsection{Preprocessing Algorithm} Let $\epsilon=0.5$ (any constant $\epsilon\in(0,1/2)$ works). 
The algorithm is described by the following pseudocode. 
\begin{tabbing}
{\bf preprocess} $G=(V,E)$\\
{\em (i) Preparing for Long-range Queries}\\
compute a $\delta$--dominating set $L$ with $\delta=\floor{\epsilon^{-1}\log n}$ as in~\cite{KP98}\\
\qquad \= (for graphs with edge weights s.t.~$\sum\limits_{e\in E}w(e)\leq O(n\log^\theta n)$, set $\delta=\floor{\epsilon^{-1}\log^{\theta+1}n}$ instead \\
\> and replace each edge $e$ by $w(e)$ edges before computing the $\delta$--dominating set)\\
for each node $l\in L$\\
\> compute Thorup's distance label~\cite{ThorupJACM04} (see Lemma~\ref{lemma:thoruplabel})\\
for each node $v\in V$\\
\qquad \= compute its nearest {\em landmark} node $l_v$ (the node $l_v\in L$ that minimizes $\dist_G(v,l_v)$)\\
             \> store $(l_v,\dist_G(v,l_v))$\\
{\em (ii) Preparing for Short-range Queries}\\
for every integer $i>0$ with $2^i\leq\ceil{2\epsilon^{-1}\delta}$\\
             \> compute a sparse neighborhood cover with radius $r=2^i$ as in~\cite{conf/podc/BuschLT07}\\
             \> let $\mathcal G^i=\{G_j^i\}$ denote this cover\\
             \> for each node $v\in V$\\
             \> \qquad \= store the list of graphs $G_j^i\in\mathcal G^i$ with $v\in V(G_j^i)$
\end{tabbing}

\paragraph{Space requirements} Each distance label has size $O(\epsilon^{-1}\log n)$ (see Lemma~\ref{lemma:thoruplabel}). 
The space requirement for the data structure computed in the first step is thus $O(n)$. 
In the second step, we iterate through $O(\log(\epsilon^{-2}\log n))=O(\log\log(n)+\log(1/\epsilon))$ levels (for weighted graphs with $\sum w(e)\leq O(n\log^{\theta}n)$ for a constant $\theta\geq0$ we have $O(\theta\log\log n+\log(1/\epsilon))$ levels). At each level $i$, for each node, we 
store a list of graphs $L^i(v)\subseteq\mathcal G^i$ of constant length $\abs{L^i(v)}=O(1)$~\cite{conf/podc/BuschLT07}. Over all levels, the space requirement is thus $O(n\log(\epsilon^{-2}\log n))$. 
Here we assume that identifiers of length $O(\log n)$~bits can be stored using constant space, which is a common assumption in the {\em word RAM model}~\cite{journals/jcss/CookR73}. 

\paragraph{Preprocessing time} The preprocessing time is dominated by the time required to compute the neighborhood covers $\sncn$, 
which is bounded by a polynomial in the number of nodes $n$~\cite{conf/podc/BuschLT07}. 
The dominating set~\cite{KP98}, the distance labels~\cite{ThorupJACM04}, and the nearest landmarks~\cite{Erwig00} can be computed in almost linear time in $n$.

\subsection{Query Algorithm} 
The algorithm is described by the following pseudocode. 

\begin{tabbing}
{\bf query} $(u,v)$\\
return the minimum of the long-range and the short-range query algorithm\\
{\em (i) Long-range Query}\\
\qquad \= return $d_G(u,l_u) + \tilde d_G(l_u, l_v) + d_G(l_v, v)$, where $\tilde d_G(\cdot,\cdot)$ is the estimate obtained from the labels\\
{\em (ii) Short-range Query}\\
binary search for a level $i$ such that $\exists G_j^i\in\mathcal G^i:u,v\in V(G_j^i)$ {\em and} $\not\exists G_{j'}^{i-1}\in\mathcal G^{i-1}:u,v\in V(G_{j'}^{i-1})$\\
\> return $2\rho2^i$, where $\rho$ is the constant for the {\em radius} in Lemma~\ref{lemma:sparsecover} (here: $\rho=24$)
\end{tabbing}

\paragraph{Running time} Computing the long-range result requires time $O(1/\epsilon)$~\cite{ThorupJACM04}. 
For the short-range pairs, 
a binary search among $O(\log(\epsilon^{-2}\log n))$ levels can be done in time $O(\log\log(\epsilon^{-2}\log n))$. 
At each search level we need to compute the intersection of two sets of constant size (recall that each node is in at most $O(1)$ graphs $G_j^i$ per level $i$~\cite{conf/podc/BuschLT07}). 

\paragraph{Stretch analysis} For any pair of nodes $(u,v)$ at distance $\dist_G(u,v)\geq\epsilon^{-1}\delta=\epsilon^{-2}\log n$, the 
long-range algorithm returns a $(1+6\epsilon)$--approximation for $\dist_G(u,v)$, since, using the triangle inequality, 
\begin{eqnarray*}
\dist_G(u,v)  &\leq& \dist_G(u,l_u) + \tilde\dist_G(l_u, l_v) + \dist_G(l_v, v)\\
&\leq&\delta + (1+\epsilon)\dist_G(l_u, l_v) + \delta\\
&\leq&\delta + (1+\epsilon)(\delta+\dist_G(u, v)+\delta) + \delta\\
&\leq&(1+\epsilon)\dist_G(u, v) + (4+2\epsilon)\delta.  
\end{eqnarray*}
For nodes at distance $\dist_G(u,v)<\epsilon^{-1}\delta$, the short-range algorithm returns a $4\rho$--approximation 
(recall that $\rho$ is the constant for the {\em radius} in~\cite{conf/podc/BuschLT07}). 
The graph $G_j^i$ with $u,v\in V(G_j^i)$ at level $i$  is a certificate that $d_G(u,v)\leq2\rho2^i$. 
Since there is no graph $G_{j'}^{i-1}$ with $u,v\in V(G_{j'}^{i-1})$ at level $i-1$, the $u$-to-$v$ distance satisfies $d_G(u,v)>2^{i-1}$. 

We conclude this section by noting that, using Lemma~\ref{lemma:sparsecover}, we have $\rho=24$ and thus the stretch of this oracle is at most 96. 
The construction of Busch et al.~\cite{conf/podc/BuschLT07} works for general minor-free graphs. There may be a construction with smaller radius for planar graphs. 

\section{$(1+\epsilon)$--Approximate Distance Oracle}
We prove the main result (Theorem~\ref{thm:eps}). 
The distance oracle presented in this section is based on the oracle with constant stretch as described in Section~\ref{sec:constantapprox}. 

\subsection{Overview}
We first run the preprocessing algorithm of Section~\ref{sec:constantapprox} with $\epsilon$ being the actual value chosen by the application divided by 6. 
Note that the only distances approximated with stretch more than $1+\epsilon$ are the ones in the range $[1,\epsilon^{-2}\log(n)]$. 
This range of logarithmic size is then handled by data structures for log-logarithmically many levels. 
At level $i$, we are interested in distances in the range $[2^i,2^{i+1})$. 

The constant-stretch distance oracle uses a very crude estimate: if two nodes are contained in the same graph $G_j^i$ with diameter $O(2^i)$ but they are not together in a graph at level $i-1$, they must be 
at distance $\Theta(2^i)$ ($O(2^i)$ due to the diameter of $G_j^i$ and $\Omega(2^i)$ since the two nodes were not together in any graph at level $i-1$). 
In the following, we provide an oracle that outputs a more precise estimate for pairs of nodes at distance $\Theta(2^i)$. 

We are aware of exact oracles for planar graphs that can answer {\em bounded-length} distance queries, which are distance queries for pairs at constant distance.
For planar graphs, Kowalik and Kurowski~\cite{journals/talg/KowalikK06} provide such a distance oracle that uses linear space (see~\cite{conf/focs/DvorakKT10} for an extension to sparse graphs). 
However, we cannot use their data structures, since 
the query time of their oracle is exponential in the length. 
In our oracle, short distances may be up to  {\em logarithmic} in~$n$. 
Another approach would be to use a distance oracle for planar graphs with bounded tree-width (see~\cite{MS12}): 
since diameter $\Theta(2^i)$ implies tree-width $w=\Theta(2^i)$~\cite{journals/algorithmica/Eppstein00,journals/algorithmica/DemaineH04}, 
the query time can be made almost proportional to $w$. However, here we aim at query time almost proportional to $1/\epsilon$ instead. 

For the $(1+\epsilon)$--stretch oracle, we introduce an additional data structure for the level graphs $G_j^i$. 
The difference to the oracle in Section~\ref{sec:constantapprox} --- one of our main technical contributions ---  
is the following data structure that can answer approximate distance queries with {\em additive stretch} $\epsilon\Delta$ for 
planar graphs with diameter~$O(\Delta)$. 

\def\rpqthm{For any integer $\Delta$, for any $\epsilon>0$, and for any planar graph on $n$ nodes with diameter $C\cdot\Delta$ for any constant $C>1$ 
there is an approximate distance oracle with {\em additive stretch}  $\epsilon\Delta$ 
using space $O(Cn\log\log(1/\epsilon)\log^*(n))$ and query time $O(C\epsilon^{-1}\log^2(1/\epsilon)\log\log(1/\epsilon)\log^*(n))$. Furthermore, this distance oracle can be computed in time $O(Cn\epsilon^{-2}\log^3(n)\log^*(n))$.}

\begin{theorem}[Additive-Stretch Approximate Distance Oracle]
\rpqthm
\label{thm:rpq}
\end{theorem}

Let us emphasize that the oracle in Theorem~\ref{thm:rpq} works for {\em edge-weighted} planar graphs as well. In the case of weighted graphs, we require that the {\em weighted} diameter 
(defined by the {\em length} of the longest shortest path) is bounded by $C\cdot\Delta$.

\subsection{High-level Construction: Preprocessing and Query Algorithms}
We now use the additive-stretch oracle as in Theorem~\ref{thm:rpq} to prove Theorem~\ref{thm:eps}, which we restate here. 

\begin{reptheorem}{thm:eps}
\mainthm
\end{reptheorem}

\begin{proof}[Proof of Theorem~\ref{thm:eps}] 
We describe the peprocessing and query algorithms. 
Since these algorithms are very similar to the algorithms in Section~\ref{sec:constantapprox}, we mainly highlight the differences. 
\paragraph{Preprocessing Algorithm} We run the preprocessing algorithm of Section~\ref{sec:constantapprox} with $\epsilon$ being the actual value chosen by the application divided by a small constant (six suffices). 
For each level $2^i$, for each graph $G_j^i$, we compute the additive-stretch oracle as in Theorem~\ref{thm:rpq}. The total space requirement per level is $\spacecomplexity$. 
The preprocessing time again depends on the time $\sncn$  required to compute the sparse neighborhood covers, since, for each level $2^i$, we compute a sparse neighborhood cover and the additive-stretch oracle in each subgraph (Theorem~\ref{thm:rpq}). 
\paragraph{Query Algorithm} We use the query algorithm of Section~\ref{sec:constantapprox} with the following adaptation. 
After the binary search for the lowest level $2^i$ with a graph $G_j^i$ that contains both $u$ and $v$, we compute the result of the additive-stretch oracle for this level {\em and also} 
for the $\ceil{\log_2(2\rho)}$ levels above ($\rho$ again denotes the radius in Lemma~\ref{lemma:sparsecover}). The reason for this is that $u$ and $v$ may be in $G_j^i$ at level $2^i$ but 
$u$ and $v$ may be at distance $cD$ for some $c\in[1,2\rho]$. At a lower level, it is not guaranteed that $G_j^i$ actually contains an approximate shortest path. 
Since $\rho=O(1)$, the query time is at most~$\querytime$. 
\end{proof}

\subsection{Polynomial Weights}
The proof of Theorem~\ref{thm:polyweights} is based on the proof of Theorem~\ref{thm:eps}. 

\begin{proof}[Proof of Theorem~\ref{thm:polyweights}] 
There are two differences to the proof of Theorem~\ref{thm:eps}. 
\paragraph{Preprocessing Algorithm} We need to consider $O(\log n)$ levels instead of $O(\log\log n)$ levels, since a node may be at distance $O(poly(n))$ from its nearest dominating node. The space complexity increases from $\spacecomplexity$ to $\spacecomplexitypoly$. We also preprocess Thorup's distance oracle for $\epsilon=1/2$ (any constant works). The space required for this is $O(n\log n)$ and thus the overall space does not increase any further. 
\paragraph{Query Algorithm} Instead of running a binary search to find the right level, we can now just query Thorup's distance oracle for constant $\epsilon$ to identify the right level. Everything else remains the same. The query time is thus $\querytimepoly$. 
\end{proof}

\section{Additive-Stretch Approximate Distance Oracle}
\label{sec:additive}

This section is devoted to the proof of Theorem~\ref{thm:rpq}. While there are distance oracles for special cases of planar graphs and
special kinds of queries~\cite{conf/stacs/DjidjevPZ95,journals/algorithmica/DjidjevPZ00,conf/icalp/DjidjevPZ91,stoc/ChenX00,journals/talg/KowalikK06,MS12}, 
we unfortunately cannot use them for our purpose. 

Let us first restate Theorem~\ref{thm:rpq}. 

\begin{reptheorem}{thm:rpq}
\rpqthm
\end{reptheorem}

\subsection{Overview}

We adapt Thorup's distance oracle as follows. During preprocessing, we recursively separate the planar graph $G$ into 
at least two pieces each with at most half the weight (Lemma~\ref{lem:thorupseparator}). Analogously to computing an 
$r$--division~\cite{journals/siamcomp/Frederickson87}, 
we  separate $G$ in a way such that both {\em (i)} the sizes of the resulting pieces are balanced and {\em (ii)} the 
boundary $\partial P$ of each piece $P$ consists of at most a constant number of shortest paths (see \cite[Section~2.5.1]{ThorupJACM04} for a similar construction). 
We stop the separation as soon as subgraphs have logarithmic size. 

There are three types of shortest $s-t$ paths we need to consider. 
\begin{enumerate}
\item Any shortest path between nodes in two different pieces must pass through a boundary path. 
\item For nodes in the same piece, the shortest path may leave through one boundary path and re-enter through another boundary path. 
\item For nodes in the same piece~$P$, the shortest path may lie entirely within $P$. 
\end{enumerate}

The third type of paths is handled by recursion. 

The first two types of paths are handled as follows. 
Instead of letting the $s-t$ path pass through any node of to boundary, that is, any node on a separator path $q\in Q$, 
we restrict the possible portals on the boundary to the $\epsilon$--cover  
for each path $Q$ (which is chosen as a set of equally-spaced nodes $C(Q)$ as in Lemma~\ref{lem:kleinsubramanian}). 
Since $C(Q)$ is an $\epsilon$--cover, the error we introduce is bounded by $\epsilon\cdot\Delta$. 

The savings in space can be obtained since {\em (i)} for each node we only need the distances to the $O(\epsilon^{-1})$ nodes 
in the covers on the boundary as opposed to the $O(\epsilon^{-1}\log n)$ nodes on all separator paths as in~\cite{ThorupJACM04} and 
{\em (ii)} we can compute these distances at query time using the Cycle MSSP data structure (Lemma~\ref{lemma:cyclemssp}). 

\subsection{Definitions}
Throughout, pieces are called $P,P',P_i,\dots$ and paths are called $Q,Q',Q_j,\dots$. 
Let $\partial P$ denote the set of separator paths on the boundary of a piece~$P$. 
By $\#\partial P$ we denote the number of shortest paths on the boundary of a piece~$P$. 
For a separator path $Q$, let $C(Q)\subseteq Q$ denote the {\em $\epsilon$--cover} of $Q$, which is a set of $O(1/\epsilon)$ equally-spaced nodes on $Q$. 
We call each node $c\in C(Q)$ a {\em portal}. 
For any piece $P$, let $\mathcal C(P)$ denote the set of portals on its boundary, $\mathcal C(P):=\bigcup\limits_{Q\in\partial P}C(Q)$. 

\subsection{Preprocessing Algorithm}
\label{sec:prepro:additive}

We first compute a decomposition we call {\em shortest-path $s$--division} (similar to an $r$--division), by repeatedly using shortest-path separators (as in Lemma~\ref{lem:thorupseparator}). 
The separator paths form a {\em decomposition tree}. Then, we compute an $\epsilon$--cover for each path. We further connect each portal 
node in an $\epsilon$--cover to all the portal nodes in covers on higher levels of the decomposition tree. 
Finally, we compute distances from nodes to their portals and store them implicitly using the Cycle MSSP data structure. 

Compared to~\cite{ThorupJACM04}, the main technical differences are: 
\begin{itemize}
\item selective storing of distances to portals\\
we store distances to portals in $\epsilon$--covers only for a restricted set of nodes
\item global $\epsilon$--covers\\
we compute one cover per path, as opposed to one cover per pair of path and node\\
(to improve query time and space requirements)
\end{itemize}

\def\numbdpaths{10}
Our preprocessing algorithm is outlined in the following pseudocode. 
\begin{figure*}[h!]
\begin{tabbing}
{\bf preprocess} $H=(V,E)$\\
let $n=\abs{V}$\\
let the set of pieces $\mathcal P=\set{V}$\\
let the {\em tree} of boundary paths $\mathcal S=\set{}$\\
\\
$(\star)$ beginning of partitioning ({\em shortest-path division}, similar to an $r$--division)\\
{\sc While} there is a piece $P\in\mathcal P$ with $\#\partial P>\numbdpaths$ or size $\abs P>\ceil{\epsilon^{-2}\log n}$\\
\qquad \= we compute a balanced separation with node weights as follows:\\
\> {\sc If} a piece has too many boundary paths, $\#\partial P>\numbdpaths$\\
\> \qquad \= assign weight 1 to each endpoint of any boundary path $Q\in\partial P$\\
\> \> assign weight 0 to all the remaining nodes\\
\> {\sc Else} \\
\> \> assign weight 1 to each node\\
\> let $T(u),T(v),T(w)$ be the three paths obtained by Lemma~\ref{lem:thorupseparator} applied to $P$, weighted as above\\
\> add $T(u),T(v),T(w)$ to the set of separator paths $\mathcal S$\\
\> remove $P$ from $\mathcal P$, partition $P$ into pieces $P_1,P_2,\dots$ and add them to $\mathcal P$\\
$(\star)$ partitioning with $O(1)$ boundary paths per region computed\\
\\
$(\dagger)$ inter-piece approximate distance oracle\\
{\sc For Each} separator path $Q\in\mathcal S$\\
\> compute an $\epsilon$--cover $C(Q)\subseteq Q$ (Lemma~\ref{lem:kleinsubramanian})\\
\> {\sc For Each} portal $c\in C(Q)$\\
\> \> compute and store the distances to all portals on {\em ancestor paths} $Q'$\\
{\sc For Each} piece $P\in\mathcal P$ and separator path $Q\in\partial P$\\
\> augment the graph induced by the piece $P$ with infinite edges such that $C(Q)$ form the outer face\\
\> compute the Cycle MSSP Data Structure for $C(Q)$ and this augmented graph (as in Lemma~\ref{lemma:cyclemssp})\\
$(\dagger)$ inter-piece oracle computed\\
\\
$(\ddagger)$ recursive call\\
{\sc For Each} piece $P\in\mathcal P$ \\
\> {\sc If} $\abs P>\ceil{\epsilon^{-2}}$\\
\> \> recurse on the graph induced by $P$\\
\> {\sc Else} \\
\> \> compute an exact distance oracle for $P$ as in Lemma~\ref{lemma:exactdq}
\end{tabbing}
\end{figure*}

\begin{claim}[Shortest-path separator $r$--division]
At step $(\star)$ we have obtained a partitioning of the vertex set $V$ into pieces $P_1,P_2,\dots$ with the following properties: 
\begin{itemize}
\item each piece $P_i$ has size $\abs{P_i}=O(\epsilon^{-1}\log n)$, 
\item the number of boundary paths $\#\partial P_i$ for each piece $P_i$ is at most ten, and 
\item the total number of pieces is at most $O(n\epsilon/\log n)$. 
\end{itemize}
\label{claim:rdiv}
\end{claim}
\begin{proof}
Our construction is similar to that of an $r$--division~\cite{journals/siamcomp/Frederickson87,nlglgn-mincut-WN10} but using 
shortest-path separators instead of cycle separators. It is also essentially the same as in~\cite[Section~2.5.1]{ThorupJACM04}; 
see~\cite[Fig.~4]{ThorupJACM04} for an illustration. Similar divisions have been used in~\cite{conf/soda/AroraGKKW98,conf/focs/GrigniKP95}. 
We recursively separate $V$ into pieces using root-paths of a single shortest-path tree as in Lemma~\ref{lem:thorupseparator}. 

Whenever a the boundary of a piece consists of more than \numbdpaths~root-paths, we separate it using Lemma~\ref{lem:thorupseparator} such that 
the boundary paths will be partitioned in a balanced way. Since we always use the same shortest-path tree $T$, we can use the following 
vertex weights: all the endpoints (leafs in $T$) of a previously selected root-path are weighted with one and all the remaining nodes are weighted 
with zero. 

The third property can be proven using the following observation. 
Let us track a piece $P$ during the execution of the preprocessing algorithm. 
Since all the separator paths are taken from the same shortest-path tree $T$, the number of boundary paths 
$\#\partial P$ increases only if $P$ is the reason why Lemma~\ref{lem:thorupseparator} is invoked. 
A piece $P'$ resulting from that invocation can have more than \numbdpaths~boundary nodes by 
inheriting $\numbdpaths-1$ paths from $P$ and gaining 3 paths from Lemma~\ref{lem:thorupseparator}. 
By one more invocation of Lemma~\ref{lem:thorupseparator}, the number of boundary paths of $P'$ 
can be reduced to less than \numbdpaths. When the recursion stops at $P'$, either $P'$ or $P$ has size $\Theta(\epsilon^{-1}\log n)$. 
\end{proof}

\begin{claim}
The total space requirement per recursion level is at most $O(n\log\log(1/\epsilon))$.
\label{claim:spaceperlevel}
\end{claim}
\begin{proof}
At each recursion level, for each portal in an $\epsilon$--cover $c\in Q$, 
the above algorithm stores $O(\epsilon^{-1}\log n)$ portals and the corresponding distances, since, 
for each portal, we store the distances to all the portals on higher levels (as in~\cite{ThorupJACM04}). 

By Claim~\ref{claim:rdiv}, the total number of pieces per level is at most $O(n\epsilon^2/\log n)$ ({\em sic!}). Each piece is surrounded by at most $\numbdpaths$ paths 
on each of which we have $O(1/\epsilon)$ portals. The total number of portals per level is thus at most $O(n\epsilon/\log n)$. 
Since each distance label requires space $O(\epsilon^{-1}\log n)$, the total space per level for this step is $O(n)$. 

At each recursion level, for each node $v\in V$, the above algorithm stores a Cycle MSSP data structure for the portals on
each boundary path (note that we can only do so since there is {\em one fixed} cover per boundary path as opposed to one cover per pair of node and boundary path). 
Since the cycle has $O(1/\epsilon)$ nodes, one Cycle MSSP data structure requires space $O(n\log\log(1/\epsilon))$ (see Lemma~\ref{lemma:cyclemssp}). 
By Claim~\ref{claim:rdiv}, each piece has at most ten boundary paths.

On the lowest level, we also store a distance oracle for a planar graph on $O(\epsilon^{-2})$ nodes. 
The overall space requirement for all these distance oracles is $O(n\log\log(1/\epsilon))$ (Lemma~\ref{lemma:exactdq}). 
\end{proof}

Note that this recursive algorithm reduces the graph size from $n\leadsto\log n$ at each level. It follows directly that the recursion depth is at most $O(\log^*n)$.
The data structure computed by the above algorithm thus occupies space $O(n\log\log(1/\epsilon)\log^*n)$.

\subsection{Query Algorithm}

The query algorithm, given a pair of nodes $u,v$, returns an estimate for $d_G(u,v)$. 
The main differences to Thorup's distance oracle~\cite{ThorupJACM04} are: 
\begin{itemize}
\item two-way approximation\\
instead of approximating the distance by $d_G(u,q)+d_Q(q,q')+d_G(q',v)$ for two portals $q,q'$ on a separator path as in Thorup's distance oracle, 
we use the estimate $d_G(u,c_u)+d_G(c_u,q)$ for a portal $c_u$ to approximate $d_G(u,q)$ and, analogously, for $d_G(q,v)$ (by doing so, each node only needs to compute and encode distances to portals on its boundary paths as opposed to {\em all} $\log n$ levels)
\item Monge search for the optimal portal\\
since we have only {\em one} $\epsilon$--cover per path, we can use the non-crossing property as in~\cite{journals/algorithmica/AggarwalKMSW87,journals/siamdm/KlaweK90,journals/jcss/FakcharoenpholR06} to compute the two-way approximation\\
we emphasize that there is also only one $\epsilon$--cover per path on higher levels (not just at the boundary of a piece) 
\end{itemize}

Let us momentarily assume that the two query nodes $u$ and $v$ are in different pieces $P_u$ and $P_v$, respectively.\footnote{Even if the two query nodes are in the same piece, it could still be that the 
shortest path passes through the separator (which consists of at most ten shortest paths). To account for these cases, we compute the minimum among paths 
passing through the separator (this computation is the same as if the nodes were in different pieces). Then we recurse until some termination criterion on the size is met.} 
Since we store distances from $u$ to its cover $\mathcal C(P_u)$ (and from $v$ to $\mathcal C(P_v)$) during preprocessing, 
we may return the minimum 
\begin{displaymath}
\min\limits_{c_u\in\mathcal C(P_u),c_v\in\mathcal C(P_v)}d_G(u,c_u) + \tilde\dist(c_u, c_v) +  d_G(c_v, v), 
\end{displaymath}
where $\tilde\dist(\cdot,\cdot)$ is a $(1+\epsilon)$--approximation for  $d_G(\cdot,\cdot)$. 
We can compute $\tilde\dist(c_u, c_v)$ for the two portals $c_u,c_v$ since they have a lowest common ancestor in the separator decomposition tree consisting of at most three shortest paths, 
to whose $\epsilon$--covers we computed and stored the shortest-path distances during preprocessing. 
\begin{claim}
The above estimate is at most  $6\epsilon\Delta$ longer than $d_G(u,v)$.
\end{claim}
\begin{proof}
Any shortest path from $u$ to $v$ must intersect the boundary of both pieces. Our approximation uses at most one additional separator path (one of the paths in the least common ancestor separator). 
Due to Lemma~\ref{lem:kleinsubramanian}, the additive distortion is at most $2\epsilon\Delta$ per path. 
\end{proof}

Per level $i$, we compute the minimum 
$\min\limits_{c_u\in\mathcal C(P_u),c_v\in\mathcal C(P_v)}d_G(u,c_u) + \tilde\dist(c_u, c_v) +  d_G(c_v, v)$. 
We know the distance from $u$ to the $O(1/\epsilon)$ portals $c_u\in\mathcal C(P_u)$ on its boundary paths. 
We also know the distance from these portals to the portals on the separator paths. 
We wish to efficiently compute the distance from $u$ to the portals on the relevant separator paths to simultaneously compute {\em all} the {\em relevant} $\tilde\dist(c_u, c_v)$. 
This computation can be done using the efficient Dijkstra implementation of~\cite{journals/jcss/FakcharoenpholR06}. 

Fakcharoenphol and Rao~\cite[Section~4.1 (and also Sec.~2.3)]{journals/jcss/FakcharoenpholR06} devised an 
ingenious implementation of Dijkstra's algorithm~\cite{Dijkstra59} 
that computes a shortest-path tree of a complete bipartite graph $H_1\cup H_2$ 
in time $O(h\log h)$ where $h=\abs{H_1}+\abs{H_2}$ (as opposed to $O(\abs{H_1}\cdot\abs{H_2})$) 
as long as the edge weights obey the Monge property (which essentially means that they correspond to distances in a planar graph 
wherein the nodes of $H_1$ and $H_2$ lie on a constant number of faces). 
We refer to this implementation as {\em FR-Dijkstra}.

The query algorithm works as described in the following pseudocode. 
\begin{tabbing}
\bf{query}$(u,v)$\\
return the minimum among the results from all recursion levels:\\
{\sc For Each} recursion level\\
\qquad \= let $u\in P_u$ and $v\in P_v$\\
\> {\sc For Each} pair of  boundary paths $(Q_u,Q_v)\in\partial P_u\times\partial P_v$\\
\> \qquad \= determine the three separator paths $Q_1,Q_2,Q_3$ that separate $Q_u$ from $Q_v$ (as in~\cite{ThorupJACM04})\\
\> \> {\sc For Each} $Q\in\set{Q_1,Q_2,Q_3}$\\
\> \> \qquad \= compute $\min\limits_{c_u\in C(Q_u)} \dist(u,q)$ for all $q\in Q$ simultaneously using FR-Dijkstra\\
\> \> \> (analogously for $v$)\\
\> \> \> compute $\min\limits_{q\in Q}\dist(u,q)+\dist(q,v)$ and keep it if it is the new minimum\\
at the lowest recursion level\\
\> {\sc If} $P_u=P_v$ {\sc Then}\\
\> \qquad \= query the exact distance oracle 
\end{tabbing}

\begin{claim}
The query algorithm runs in time $O(\epsilon^{-1}\log^2(1/\epsilon)\log\log(1/\epsilon)\log^*(n))$. 
\end{claim}
\begin{proof}
We compute the distances from $u$ to the portals on the boundary path $C(Q_u)$ using the Cycle MSSP data structure 
in time $O(\epsilon^{-1}\log^2(1/\epsilon)\log\log(1/\epsilon))$  (Lemma~\ref{lemma:cyclemssp}).

For each pair of  boundary paths $(Q_u,Q_v)\in\partial P_u\times\partial P_v$ there is a constant 
number of separator paths $Q$ on higher levels we need to consider (see constant query time 
oracle in~\cite{ThorupJACM04}). 

Per separator path $Q$ the FR-Dijkstra algorithm runs in time 
$O((\abs{C(Q_u)}+\abs{C(Q)})\log(\abs{C(Q_u)}+\abs{C(Q)}))=O(\epsilon^{-1}\log(1/\epsilon))$. 
(Since we do not actually need the on-line version of the bipartite Monge search in~\cite{journals/jcss/FakcharoenpholR06}, 
the $\log(1/\epsilon)$ factor can be eliminated by using~\cite{journals/algorithmica/AggarwalKMSW87,journals/siamdm/KlaweK90} instead. 
However, we use FR-Dijkstra within the Cycle MSSP data structure, whose query time currently dominates the overall query time anyway.)

This search is done for all the $O(\log^*n)$ levels of the recursion. 
On the lowest level, we also query the exact distance oracle in time $O(\epsilon^{-1}\log^2(1/\epsilon)\log\log(1/\epsilon))$ (Lemma~\ref{lemma:exactdq}). 
\end{proof}

\section{Conclusion}
Our $(1+\epsilon)$--approximate distance oracle for planar graphs has a better space-time tradeoff --- both in terms of $n$ and $\epsilon$ ---  
than previous oracles. 
The improved tradeoff currently comes at a cost:  the oracle cannot be distributed as a labeling scheme 
(it is however not clear whether $o(\log^2n)$~bit approximate distance labels for planar graphs exist at all) and, 
for the improvement in terms of $n$, there is a dependency on the largest edge weights (Thorup's oracle depends on 
the largest integer weight for digraphs only). 
We believe that these sacrifices may be worthwhile since they allow us to achieve almost linear dependency on $1/\epsilon$ and 
remove the notorious logarithmic dependency on $n$, contributing to 
an important next step towards a linear-space approximate distance oracle with almost constant query time. 

An intermediate goal could be to find a linear-space distance oracle with arbitrary constant stretch and (almost) constant query time. 
For our construction, the overhead required to transform a constant-stretch oracle into a $(1+\epsilon)$--stretch oracle is proportional to only $\log^*(n)$ (and some overhead in~$\log(1/\epsilon)$). 
Although our construction of the $(1+\epsilon)$--stretch oracle depends on the specific features of the constant-stretch oracle, given a better constant-stretch approximate distance oracle, 
it may be quite possible to improve our $(1+\epsilon)$--stretch oracle as well. 

\section*{Acknowledgments}
Many thanks to Mikkel Thorup for encouraging our research on these tradeoffs and for the valuable comments on an earlier version of this work.

\bibliographystyle{alpha}
\bibliography{planaradq,../../thesis/related}

\newcommand{\etalchar}[1]{$^{#1}$}
\begin{thebibliography}{AGMW07}

\bibitem[ABCP98]{journals/siamcomp/AwerbuchBCP98}
Baruch Awerbuch, Bonnie Berger, Lenore Cowen, and David Peleg.
\newblock Near-linear time construction of sparse neighborhood covers.
\newblock {\em SIAM Journal on Computing}, 28(1):263--277, 1998.
\newblock Announced at FOCS 1993.

\bibitem[ACC{\etalchar{+}}96]{ArikatiCCDSZ96}
Srinivasa~Rao Arikati, Danny~Z. Chen, L.~Paul Chew, Gautam Das, Michiel H.~M.
  Smid, and Christos~D. Zaroliagis.
\newblock Planar spanners and approximate shortest path queries among obstacles
  in the plane.
\newblock In {\em Algorithms - ESA '96, Fourth Annual European Symposium,
  Barcelona, Spain, September 25-27, 1996, Proceedings}, pages 514--528, 1996.

\bibitem[AGK{\etalchar{+}}98]{conf/soda/AroraGKKW98}
Sanjeev Arora, Michelangelo Grigni, David~R. Karger, Philip~N. Klein, and
  Andrzej Woloszyn.
\newblock A polynomial-time approximation scheme for weighted planar graph
  {TSP}.
\newblock In {\em SODA}, pages 33--41, 1998.

\bibitem[AGMW07]{conf/spaa/AbrahamGMW07}
Ittai Abraham, Cyril Gavoille, Dahlia Malkhi, and Udi Wieder.
\newblock Strong-diameter decompositions of minor free graphs.
\newblock In {\em SPAA 2007: Proceedings of the 19th Annual ACM Symposium on
  Parallel Algorithms and Architectures, San Diego, California, USA, June 9-11,
  2007}, pages 16--24, 2007.

\bibitem[AKM{\etalchar{+}}87]{journals/algorithmica/AggarwalKMSW87}
Alok Aggarwal, Maria~M. Klawe, Shlomo Moran, Peter~W. Shor, and Robert~E.
  Wilber.
\newblock Geometric applications of a matrix-searching algorithm.
\newblock {\em Algorithmica}, 2:195--208, 1987.
\newblock Announced at SoCG 1986.

\bibitem[AP90]{conf/focs/AwerbuchP90a}
Baruch Awerbuch and David Peleg.
\newblock Sparse partitions (extended abstract).
\newblock In {\em 31st Annual Symposium on Foundations of Computer Science,
  22-24 October 1990, St. Louis, Missouri, USA}, pages 503--513, 1990.

\bibitem[BBJ{\etalchar{+}}02]{conf/esa/BarrettBJKM02}
Christopher~L. Barrett, Keith~R. Bisset, Riko Jacob, Goran Konjevod, and
  Madhav~V. Marathe.
\newblock Classical and contemporary shortest path problems in road networks:
  Implementation and experimental analysis of the {TRANSIMS} router.
\newblock In {\em Algorithms - ESA 2002, 10th Annual European Symposium, Rome,
  Italy, September 17-21, 2002, Proceedings}, pages 126--138, 2002.

\bibitem[BG07]{SPAcceleration}
Zachary~K. Baker and Maya Gokhale.
\newblock On the acceleration of shortest path calculations in transportatoin
  networks.
\newblock In {\em International Symposium on Field-Programmable Custom
  Computing Machines}, pages 23--32, 2007.

\bibitem[BLT07]{conf/podc/BuschLT07}
Costas Busch, Ryan LaFortune, and Srikanta Tirthapura.
\newblock Improved sparse covers for graphs excluding a fixed minor.
\newblock In {\em Proceedings of the Twenty-Sixth Annual ACM Symposium on
  Principles of Distributed Computing, PODC 2007, Portland, Oregon, USA, August
  12-15, 2007}, pages 61--70, 2007.

\bibitem[Cab06]{conf/soda/Cabello06}
Sergio Cabello.
\newblock Many distances in planar graphs.
\newblock In {\em Proceedings of the Seventeenth Annual ACM-SIAM Symposium on
  Discrete Algorithms, SODA 2006, Miami, Florida, USA, January 22-26, 2006},
  pages 1213--1220, 2006.
\newblock A preprint of the journal version is available in the University of
  Ljubljana preprint series, Vol. 47 (2009), 1089.

\bibitem[CCRK04]{conf/icdcs/CostaCRK04}
Manuel Costa, Miguel Castro, Antony I.~T. Rowstron, and Peter~B. Key.
\newblock Pic: Practical internet coordinates for distance estimation.
\newblock In {\em 24th International Conference on Distributed Computing
  Systems (ICDCS 2004), 24-26 March 2004, Hachioji, Tokyo, Japan}, pages
  178--187, 2004.

\bibitem[Che95]{conf/soda/Chen95}
Danny~Ziyi Chen.
\newblock On the all-pairs euclidean short path problem.
\newblock In {\em Proceedings of the Sixth Annual ACM-SIAM Symposium on
  Discrete Algorithms}, pages 292--301, 1995.

\bibitem[CR73]{journals/jcss/CookR73}
Stephen~A. Cook and Robert~A. Reckhow.
\newblock Time bounded random access machines.
\newblock {\em Journal of Computer and System Sciences}, 7(4):354--375, 1973.
\newblock Announced at STOC 1972.

\bibitem[CX00]{stoc/ChenX00}
Danny~Z. Chen and Jinhui Xu.
\newblock Shortest path queries in planar graphs.
\newblock In {\em Proceedings of the ACM Symposium on Theory of Computing
  (STOC)}, pages 469--478, 2000.

\bibitem[CZE{\etalchar{+}}11]{CausalRegulatoryNetworks}
Leonid Chindelevitch, Daniel Ziemek, Ahmed Enayetallah, Ranjit Randhawa, Ben
  Sidders, Christoph Brockel, and Enoch Huang.
\newblock Causal reasoning on biological networks: Interpreting transcriptional
  changes.
\newblock In {\em Research in Computational Molecular Biology, 15th Annual
  International Conference, RECOMB 2011}, 2011.

\bibitem[DCKM04]{1015471}
Frank Dabek, Russ Cox, Frans Kaashoek, and Robert Morris.
\newblock Vivaldi: a decentralized network coordinate system.
\newblock In {\em SIGCOMM '04: Proceedings of the 2004 conference on
  Applications, technologies, architectures, and protocols for computer
  communications}, pages 15--26, 2004.

\bibitem[DF85]{DyerFrieze85}
Martin~E. Dyer and Alan~M. Frieze.
\newblock A simple heuristic for the $p$--centre problem.
\newblock {\em Operations Research Letters}, 3(6):285--288, 1985.

\bibitem[DGJ08]{DIMACSChallenge}
Camil Demetrescu, Andrew~V. Goldberg, and David~S. Johnson.
\newblock Implementation challenge for shortest paths.
\newblock In {\em Encyclopedia of Algorithms}. 2008.

\bibitem[DH04]{journals/algorithmica/DemaineH04}
Erik~D. Demaine and Mohammad~Taghi Hajiaghayi.
\newblock Diameter and treewidth in minor-closed graph families, revisited.
\newblock {\em Algorithmica}, 40(3):211--215, 2004.

\bibitem[Die05]{DiestelBook}
Reinhard Diestel.
\newblock {\em Graph Theory (Graduate Texts in Mathematics)}.
\newblock Springer, August 2005.

\bibitem[Dij59]{Dijkstra59}
Edsger~Wybe Dijkstra.
\newblock A note on two problems in connexion with graphs.
\newblock {\em Numerische Mathematik}, 1:269--271, 1959.

\bibitem[Dji96]{DjidjevWG96}
Hristo~Nikolov Djidjev.
\newblock Efficient algorithms for shortest path problems on planar digraphs.
\newblock In {\em Graph-Theoretic Concepts in Computer Science, 22nd
  International Workshop, WG '96, Cadenabbia (Como), Italy, June 12-14, 1996,
  Proceedings}, pages 151--165, 1996.

\bibitem[DKT10]{conf/focs/DvorakKT10}
Zdenek Dvorak, Daniel Kr{\'a}l, and Robin Thomas.
\newblock Deciding first-order properties for sparse graphs.
\newblock In {\em 51th Annual IEEE Symposium on Foundations of Computer
  Science, FOCS 2010, October 23-26, 2010, Las Vegas, Nevada, USA}, pages
  133--142, 2010.

\bibitem[DPZ91]{conf/icalp/DjidjevPZ91}
Hristo~Nikolov Djidjev, Grammati~E. Pantziou, and Christos~D. Zaroliagis.
\newblock Computing shortest paths and distances in planar graphs.
\newblock In {\em Automata, Languages and Programming, 18th International
  Colloquium, ICALP91, Madrid, Spain, July 8-12, 1991, Proceedings}, pages
  327--338, 1991.

\bibitem[DPZ95]{conf/stacs/DjidjevPZ95}
Hristo~Nikolov Djidjev, Grammati~E. Pantziou, and Christos~D. Zaroliagis.
\newblock On-line and dynamic algorithms for shorted path problems.
\newblock In {\em STACS}, pages 193--204, 1995.

\bibitem[DPZ00]{journals/algorithmica/DjidjevPZ00}
Hristo~Nikolov Djidjev, Grammati~E. Pantziou, and Christos~D. Zaroliagis.
\newblock Improved algorithms for dynamic shortest paths.
\newblock {\em Algorithmica}, 28(4):367--389, 2000.

\bibitem[DSSW09]{RoutePlanningSurvey}
Daniel Delling, Peter Sanders, Dominik Schultes, and Dorothea Wagner.
\newblock Engineering route planning algorithms.
\newblock In {\em Algorithmics of Large and Complex Networks - Design,
  Analysis, and Simulation [DFG priority program 1126]}, pages 117--139, 2009.

\bibitem[EBN09]{conf/infocom/ErikssonBN09}
Brian Eriksson, Paul Barford, and Robert~D. Nowak.
\newblock Estimating hop distance between arbitrary host pairs.
\newblock In {\em INFOCOM 2009. 28th IEEE International Conference on Computer
  Communications, Joint Conference of the IEEE Computer and Communications
  Societies, 19-25 April 2009, Rio de Janeiro, Brazil}, pages 801--809, 2009.

\bibitem[Epp00]{journals/algorithmica/Eppstein00}
David Eppstein.
\newblock Diameter and treewidth in minor-closed graph families.
\newblock {\em Algorithmica}, 27(3-4):275--291, 2000.

\bibitem[Erw00]{Erwig00}
Martin Erwig.
\newblock The graph {Voronoi} diagram with applications.
\newblock {\em Networks}, 36(3):156--163, 2000.

\bibitem[FJ89]{journals/siamcomp/FredericksonJ89}
Greg~N. Frederickson and Ravi Janardan.
\newblock Efficient message routing in planar networks.
\newblock {\em SIAM J. Comput.}, 18(4):843--857, 1989.

\bibitem[FJ90]{journals/siamcomp/FredericksonJ90}
Greg~N. Frederickson and Ravi Janardan.
\newblock Space-efficient message routing in $c$-decomposable networks.
\newblock {\em SIAM J. Comput.}, 19(1):164--181, 1990.

\bibitem[FR06]{journals/jcss/FakcharoenpholR06}
Jittat Fakcharoenphol and Satish Rao.
\newblock Planar graphs, negative weight edges, shortest paths, and near linear
  time.
\newblock {\em Journal of Computer and System Sciences}, 72(5):868--889, 2006.
\newblock Announced at FOCS 2001.

\bibitem[Fre87]{journals/siamcomp/Frederickson87}
Greg~N. Frederickson.
\newblock Fast algorithms for shortest paths in planar graphs, with
  applications.
\newblock {\em SIAM Journal on Computing}, 16(6):1004--1022, 1987.

\bibitem[GKK{\etalchar{+}}01]{conf/esa/GavoilleKKPP01}
Cyril Gavoille, Michal Katz, Nir~A. Katz, Christophe Paul, and David Peleg.
\newblock Approximate distance labeling schemes.
\newblock In {\em Algorithms - ESA 2001, 9th Annual European Symposium, Aarhus,
  Denmark, August 28-31, 2001, Proceedings}, pages 476--487, 2001.

\bibitem[GKP95]{conf/focs/GrigniKP95}
Michelangelo Grigni, Elias Koutsoupias, and Christos~H. Papadimitriou.
\newblock An approximation scheme for planar graph tsp.
\newblock In {\em FOCS}, pages 640--645, 1995.

\bibitem[GKR01]{conf/focs/GuptaKR01}
Anupam Gupta, Amit Kumar, and Rajeev Rastogi.
\newblock Traveling with a pez dispenser (or, routing issues in mpls).
\newblock In {\em FOCS}, pages 148--157, 2001.

\bibitem[Gol07]{conf/sofsem/Goldberg07}
Andrew~V. Goldberg.
\newblock Point-to-point shortest path algorithms with preprocessing.
\newblock In {\em SOFSEM 2007: Theory and Practice of Computer Science, 33rd
  Conference on Current Trends in Theory and Practice of Computer Science,
  Harrachov, Czech Republic, January 20-26, 2007, Proceedings}, pages 88--102,
  2007.

\bibitem[GPPR04]{journals/jal/GavoillePPR04}
Cyril Gavoille, David Peleg, St{\'e}phane P{\'e}rennes, and Ran Raz.
\newblock Distance labeling in graphs.
\newblock {\em J. Algorithms}, 53(1):85--112, 2004.
\newblock Announced at SODA 2001.

\bibitem[HKRS97]{journals/jcss/HenzingerKRS97}
Monika~Rauch Henzinger, Philip~Nathan Klein, Satish Rao, and Sairam
  Subramanian.
\newblock Faster shortest-path algorithms for planar graphs.
\newblock {\em Journal of Computer and System Sciences}, 55(1):3--23, 1997.
\newblock Announced at STOC 1994.

\bibitem[JHR96]{conf/cikm/JingHR96}
Ning Jing, Yun-Wu Huang, and Elke~A. Rundensteiner.
\newblock Hierarchical optimization of optimal path finding for transportation
  applications.
\newblock In {\em CIKM '96, Proceedings of the Fifth International Conference
  on Information and Knowledge Management, November 12 - 16, 1996, Rockville,
  Maryland, USA}, pages 261--268, 1996.

\bibitem[Kar29]{Karinthy}
Frigyes Karinthy.
\newblock {\em Lancszemek}.
\newblock 1929.

\bibitem[KK90]{journals/siamdm/KlaweK90}
Maria~M. Klawe and Daniel~J. Kleitman.
\newblock An almost linear time algorithm for generalized matrix searching.
\newblock {\em SIAM J. Discrete Math.}, 3(1):81--97, 1990.

\bibitem[KK06]{journals/talg/KowalikK06}
Lukasz Kowalik and Maciej Kurowski.
\newblock Oracles for bounded-length shortest paths in planar graphs.
\newblock {\em ACM Transactions on Algorithms}, 2(3):335--363, 2006.
\newblock Announced at STOC 2003.

\bibitem[KKS11]{KKS}
Kenichi Kawarabayashi, Philip~Nathan Klein, and Christian Sommer.
\newblock Linear-space approximate distance oracles for planar, bounded-genus,
  and minor-free graphs.
\newblock In {\em Automata, Languages and Programming, 38th International
  Colloquium (ICALP)}, pages 135--146, 2011.

\bibitem[Kle02]{conf/soda/Klein02}
Philip~Nathan Klein.
\newblock Preprocessing an undirected planar network to enable fast approximate
  distance queries.
\newblock In {\em Proceedings of the Thirteenth Annual ACM-SIAM Symposium on
  Discrete Algorithms, January 6-8, 2002, San Francisco, CA, USA}, pages
  820--827, 2002.

\bibitem[Kle05]{conf/soda/Klein05}
Philip~Nathan Klein.
\newblock Multiple-source shortest paths in planar graphs.
\newblock In {\em Proceedings of the Sixteenth Annual ACM-SIAM Symposium on
  Discrete Algorithms, SODA 2005, Vancouver, British Columbia, Canada, January
  23-25, 2005}, pages 146--155, 2005.

\bibitem[KLMN04]{conf/focs/KrauthgamerLMN04}
Robert Krauthgamer, James~R. Lee, Manor Mendel, and Assaf Naor.
\newblock Measured descent: A new embedding method for finite metrics.
\newblock In {\em 45th Symposium on Foundations of Computer Science (FOCS
  2004), 17-19 October 2004, Rome, Italy, Proceedings}, pages 434--443, 2004.

\bibitem[KP98]{KP98}
Shay Kutten and David Peleg.
\newblock Fast distributed construction of small $k$-dominating sets and
  applications.
\newblock {\em Journal of Algorithms}, 28(1):40--66, 1998.
\newblock Announced at PODC 1995.

\bibitem[KS98]{journals/algorithmica/KleinS98}
Philip~Nathan Klein and Sairam Subramanian.
\newblock A fully dynamic approximation scheme for shortest paths in planar
  graphs.
\newblock {\em Algorithmica}, 22(3):235--249, 1998.

\bibitem[LT79]{LT79}
Richard~J. Lipton and Robert~Endre Tarjan.
\newblock A separator theorem for planar graphs.
\newblock {\em SIAM Journal on Applied Mathematics}, 36(2):177--189, 1979.

\bibitem[Mil67]{Milgram}
Stanley Milgram.
\newblock The small world problem.
\newblock {\em Psychology Today}, 1:61--67, 1967.

\bibitem[Mil86]{journals/jcss/Miller86}
Gary~L. Miller.
\newblock Finding small simple cycle separators for 2-connected planar graphs.
\newblock {\em Journal of Computer and System Sciences}, 32(3):265--279, 1986.
\newblock Announced at STOC 1984.

\bibitem[MS12]{MS12}
Shay Mozes and Christian Sommer.
\newblock Exact distance oracles for planar graphs.
\newblock In {\em Proceedings of the 23rd ACM-SIAM Symposium on Discrete
  Algorithms}, 2012.
\newblock to appear, preprint available on the arXiv
  \url{http://arxiv.org/abs/1011.5549}.

\bibitem[New01]{ScientificCollaborationNetworks}
Mark E.~J. Newman.
\newblock Scientific collaboration networks. {II}. shortest paths, weighted
  networks, and centrality.
\newblock {\em Physical Review E (Statistical, Nonlinear, and Soft Matter
  Physics)}, 64, 2001.

\bibitem[Nus11]{Nussbaum11}
Yahav Nussbaum.
\newblock Improved distance queries in planar graphs.
\newblock In {\em 12th International Symposium on Algorithms and Data
  Structures (WADS)}, pages 642--653, 2011.

\bibitem[NZ02]{NgZhang}
T.~S.~Eugene Ng and Hui Zhang.
\newblock Predicting internet network distance with coordinates-based
  approaches.
\newblock In {\em INFOCOM}, 2002.

\bibitem[Pel00]{PelegLabeling00}
David Peleg.
\newblock Proximity-preserving labeling schemes.
\newblock {\em J. Graph Theory}, 33(3):167--176, 2000.
\newblock Announced at WG 1999.

\bibitem[Ple80]{PlesnikCentersComplexity}
J{\'a}n Plesn\'{\i}k.
\newblock On the computational complexity of centers locating in a graph.
\newblock {\em Applications of Mathematics}, 25(6):445--452, 1980.

\bibitem[Ple87]{PlesnikPCenters}
J{\'a}n Plesn\'{\i}k.
\newblock A heuristic for the $p$--center problems in graphs.
\newblock {\em Discrete Applied Mathematics}, 17(3):263--268, 1987.

\bibitem[PR10]{PatrascuRoditty}
Mihai Patrascu and Liam Roditty.
\newblock Distance oracles beyond the {Thorup--Zwick} bound.
\newblock In {\em 51st Annual IEEE Symposium on Foundations of Computer
  Science, FOCS 2010, October 23-26, 2010, Las Vegas, Nevada, USA}, 2010.

\bibitem[RN04]{journals/fgcs/RaneyN04}
Bryan Raney and Kai Nagel.
\newblock Iterative route planning for large-scale modular transportation
  simulations.
\newblock {\em Future Generation Computer Systems}, 20(7):1101--1118, 2004.

\bibitem[SS80]{1094690}
Mischa Schwartz and Thomas~E. Stern.
\newblock Routing techniques used in computer communication networks.
\newblock {\em IEEE Transactions on Communications}, 28(4):539--552, Apr 1980.

\bibitem[ST08]{Shavitt:2008:HEI:1373452.1373455}
Yuval Shavitt and Tomer Tankel.
\newblock Hyperbolic embedding of internet graph for distance estimation and
  overlay construction.
\newblock {\em IEEE/ACM Trans. Netw.}, 16:25--36, February 2008.

\bibitem[Sto99]{SmartMove}
Bryan Stout.
\newblock Smart move: Intelligent path-finding.
\newblock Online at
  \href{http://www.gamasutra.com/view/feature/3317/smart\_move\_intelligent\_.%
php}{gamasutra.com/view/feature/3317/smart\_move\_intelligent\_.php}, 1999.

\bibitem[STW04]{conf/edbt/SchenkelTW04}
Ralf Schenkel, Anja Theobald, and Gerhard Weikum.
\newblock {HOPI}: An efficient connection index for complex {XML} document
  collections.
\newblock In {\em Advances in Database Technology - EDBT 2004, 9th
  International Conference on Extending Database Technology, Heraklion, Crete,
  Greece, March 14-18, 2004, Proceedings}, pages 237--255, 2004.

\bibitem[STW05]{conf/icde/SchenkelTW05}
Ralf Schenkel, Anja Theobald, and Gerhard Weikum.
\newblock Efficient creation and incremental maintenance of the {HOPI} index
  for complex {XML} document collections.
\newblock In {\em Proceedings of the 21st International Conference on Data
  Engineering, ICDE 2005, 5-8 April 2005, Tokyo, Japan}, pages 360--371, 2005.

\bibitem[SVY09]{SparseDO}
Christian Sommer, Elad Verbin, and Wei Yu.
\newblock Distance oracles for sparse graphs.
\newblock In {\em 50th Annual IEEE Symposium on Foundations of Computer Science
  (FOCS)}, pages 703--712, 2009.

\bibitem[Tho04]{ThorupJACM04}
Mikkel Thorup.
\newblock Compact oracles for reachability and approximate distances in planar
  digraphs.
\newblock {\em Journal of the ACM}, 51(6):993--1024, 2004.
\newblock Announced at FOCS 2001.

\bibitem[TZ05]{ThorupZwick2005}
Mikkel Thorup and Uri Zwick.
\newblock Approximate distance oracles.
\newblock {\em Journal of the ACM}, 52(1):1--24, 2005.
\newblock Announced at STOC 2001.

\bibitem[Ung51]{Ungar01101951}
Peter Ungar.
\newblock A theorem on planar graphs.
\newblock {\em Journal of the London Mathematical Society}, s1-26(4):256--262,
  1951.

\bibitem[WN10a]{WNSubQ10}
Christian Wulff-Nilsen.
\newblock Constant time distance queries in planar unweighted graphs with
  subquadratic preprocessing time, 2010.
\newblock Computational Geometry, special issue on the 25th European Workshop
  on Computational Geometry (to appear).

\bibitem[WN10b]{nlglgn-mincut-WN10}
Christian Wulff-Nilsen.
\newblock Min $st$-cut of a planar graph in {$O(n \log\log n)$} time.
\newblock {\em CoRR}, abs/1007.3609, 2010.

\bibitem[Zar08]{reference/algo/Zaroliagis}
Christos Zaroliagis.
\newblock Engineering algorithms for large network applications.
\newblock In {\em Encyclopedia of Algorithms}. 2008.

\bibitem[ZKM97]{Ziliaskopoulos:Kotzinos:1997}
Athanasios~K. Ziliaskopoulos, Dimitri Kotzinos, and Hani~S. Mahmassani.
\newblock Design and implementation of parallel time-dependent least time path
  algorithms for intelligent transportation systems applications.
\newblock {\em Transportation Research Part C: Emerging Technologies},
  5(2):95--107, 1997.

\end{thebibliography}

\end{document}